\documentclass[aps,pra,twocolumn,showpacs,superscriptaddress]{revtex4-1}
\usepackage{amsmath,amssymb} %
\usepackage{bm,braket} %
\usepackage{amsthm}
 \newtheorem*{theo}{Theorem} %
 \newtheorem*{theo-a}{Theorem$'$} %
 \newtheorem*{prop}{Proposition} %
\theoremstyle{definition}
 \newtheorem{defi}{Definition} %

\newtheoremstyle{defiapp}% name of the style to be used
{}% measure of space to leave above the theorem. E.g.: 3pt
{}% measure of space to leave below the theorem. E.g.: 3pt
{}% name of font to use in the body of the theorem
{}% measure of space to indent
{\bfseries}% name of head font
{.}% punctuation between head and body
{5pt}% space after theorem head (default: 5pt)
{\thmname{#1}\thmnumber{ #2$'$}\thmnote{\hspace{2pt}(#3)}}% theorem head spec
\theoremstyle{defiapp}
\newtheorem{defi-a}{Definition}

\usepackage{hyperxmp} % XMP support with hyperref
\usepackage[pdfencoding=auto,unicode,pdfa]{hyperref} % PDF/A compatible
\hypersetup{
    pdflang={en}
 }
\DeclareMathOperator{\tr}{tr} %
\newcommand{\bN}{\mathbb{N}} %
\newcommand{\bR}{\mathbb{R}} %
\newcommand{\bC}{\mathbb{C}} %
\newcommand{\incon}{?} %
\newcommand{\incont}{``?''} %
\newcommand{\Prob}{\mathrm{P}} %
\newcommand{\inst}[1]{\mathcal{#1}}

\newcommand{\hil}[1]{\mathcal{#1}}
\newcommand{\hilh}{\hil{H}} %

\newcommand{\bof}[1]{B(\hil{#1})} %
\newcommand{\bofh}{\bof{H}}

\newcommand{\bone}[1]{B^1(\hil{#1})} %
\newcommand{\bonen}{B^1} %
\newcommand{\boneh}{\bone{H}}

\newcommand{\jincon}{J\cup\{\incon\}}
\newcommand{\Omegaa}{{\jincon}}

\begin{document}

\title{Disturbance by optimal discrimination}
\author{Ry\^uitir\^o Kawakubo}
\email{rkawakub@rk.phys.keio.ac.jp}
\affiliation{Department of Physics, Keio University, Yokohama 223-8522, Japan}
\author{Tatsuhiko Koike}
\email{koike@phys.keio.ac.jp}
\affiliation{Department of Physics, Keio University, Yokohama 223-8522, Japan}
\date{\today}

\begin{abstract}
We discuss the disturbance by measurements which unambiguously
discriminate between given candidate states. We prove that such an
optimal measurement necessarily changes distinguishable states
indistinguishable when the inconclusive outcome is obtained. The
result was previously shown by Chefles~[Phys. Lett. A 239, 339 (1998)]
under restrictions on the class of quantum measurements and on the
definition of optimality. Our theorems remove these restrictions and
are also applicable to infinitely many candidate states. 
Combining with our previous results, one can obtain concrete
mathematical conditions for the resulting states.
%resulting from optimal measurements. 
The method
may have a wide variety of applications in contexts other than state
discrimination.   
\end{abstract}

\maketitle

\section{Introduction}

Optimal quantum measurements play fundamental roles in 
many subjects in 
quantum foundations and quantum information. 
The subjects include 
error-disturbance relations~\cite{Ozawa04}, 
quantum coding theorems~\cite{Hol98}, 
entanglement distillation~\cite{Bennett96a} 
%%,Bennett96b,Bennett96c}, 
state estimation~\cite{Helstrom76}, 
state discrimination~\cite{Chefles00}, 
state protection~\cite{Wakamura17a}, etc.
Although it is sometimes difficult to obtain concrete forms of optimal
measurements, 
the characterization of those 
%%optimal 
measurements that 
%%where the optimal measurements which 
maximize the precision, 
maximize the transmission rate, etc., 
gives fundamental bounds in quantum mechanics and on efficiency of
various information processing. 
In this paper, we discuss a natural and intuitive
principle in optimality: 
The measurement which does ``something'' the best leaves no room for
doing ``something'' afterwards. 
Intuitively, if the room remained, one would be able to 
improve the
original measurement by composing that with a measurement 
which does an amount of ``something'' later. 
This simple reasoning seems to have a wide scope.  
%%%%%%%%%%%%%%%%%%
%It may lead to 
%an abstract theory on optimal measurements 
%or may be applied to 
%concrete subjects 
%above. 
%However, it is not necessarily 
%trivial to implement the idea as a rigorous statement 
%in general or in each subject. 
%One must carefully choose the definition of optimality, 
%assumptions and conclusions. 
%Conversely, this will make the original problem more transparent. 
%%%%%%%%%%%%%%%%%%
%It may lead to 
%an abstract theory on optimal measurements 
%or may be applied to 
%concrete subjects 
%above. 
However, it is not necessarily 
trivial to implement the idea as a rigorous statement 
in general or in each subject. 
One must carefully choose the definition of optimality, 
assumptions and conclusions. 
%Conversely, this will make the original problem more transparent. 
%%%%%%%%%%%%%%%%%%
We will demonstrate, 
as a prototype, 
that the principle works well and makes the original problem more transparent
in the context of unambiguous state discrimination, 
where we also 
notice subtleties in the application 
thereof. 
The principle %greatly 
generalizes the previous results  
and simplifies the proof, 
which may not be achieved by 
other approaches such as 
extremity (in the mathematical sense) 
of the optimal measurements, 
even if we restrict ourselves to convex evaluation functions.

Unambiguous discrimination
is one of possible frameworks for state discrimination.
There,
one must answer the correct input state among the candidates after performing a quantum measurement. 
One must not
take a state for another, 
though 
one 
can 
answer ``inconclusive'' or ``?.'' 
Unambiguous discrimination between two states 
was 
introduced by Ivanovic~\cite{Ivanovic1987257}
and developed in Refs.~\cite{Dieks1988303,Peres198819,Jaeger199583}.
Chefles~\cite{Chefles1998339} showed 
that finitely many pure states 
are distinguishable 
if and only if they
are linearly independent. 
Feng \textit{et.~al.}~\cite{PhysRevA.70.012308} extended the result to
mixed states. 
There are also interesting examples where infinitely many
candidate states are involved. 
An example is 
the set of coherent states corresponding to 
a lattice in the classical phase space,
which was considered 
by von Neumann~\cite{vonneumann32} 
in the context of simultaneous measurements of 
position and momentum. 
The present authors~\cite{1751-8121-49-26-265201} 
generalized the results above on unambiguous discrimination to 
infinitely many candidate states. 
The application 
to von Neumann's lattice 
led to a natural characterization of Planck's constant 
from the viewpoint of state discrimination.

The 
studies of unambiguous discrimination above 
mainly discuss its possibility and accuracy. 
Not much attention was paid to 
the disturbance caused by
unambiguous discrimination measurements. 
One exception is a part of 
Chefles's work~\cite{Chefles1998339}.
He showed,
under some restrictions explained below, 
that optimal discrimination measurements 
change the input states
to linearly dependent (and thus indistinguishable) ones 
if the inconclusive outcome {\incont} is obtained. 
The claim is interesting 
%%in the sense that 
because 
it concerns about disturbance of an optimal measurement.
However, besides the finiteness of the candidate states, 
the results obtained there were restrictive 
in the following sense. 
First, 
the states to be discriminated were assumed to be pure. 
Second, 
not all measurements allowed by quantum mechanics were considered. 
The measurements were 
restricted to those 
which change pure states to pure states when the outcome is
inconclusive. 
It is quite common, however,  that the output state is mixed even if
 the input is pure. 
Third, 
a particular evaluation function was considered to define 
the optimality. 
Namely, 
existence of a prior probability distribution is assumed and 
the average success probability was chosen. 
Even in the unambiguous discrimination between finite number of states, 
it may be as natural, 
for example, 
to define the optimality by 
maximization of 
the minimum success probability 
taken over the candidate states. 

In this paper,
we show that
optimal unambiguous discrimination measurements 
make distinguishable states indistinguishable
provided that the outcome is the inconclusive one. 
We apply the simple principle 
described at the beginning of this paper and 
derive the result directly from the definition of the optimality, 
not resorting to 
the detailed mathematical properties of the states. 
The results are free from all restrictions mentioned above, 
and 
is valid for mixed states, for general measurements, 
for a wide class of evaluation functions 
(which are not necessarily convex), 
and for 
infinitely many candidate states. 
%%Note that infinitely many states require 
A careful treatment 
is necessary for 
infinitely many states 
since 
distinguishability naturally splits into slightly different levels, 
distinguishability and 
uniform distinguishability~\cite{1751-8121-49-26-265201}. 
One can obtain 
detailed 
mathematical characterization of the states resulting 
from optimal measurements 
by 
combining the results in 
Ref.~\cite{1751-8121-49-26-265201} 
and those in the present work.

The paper is organized as follows.
In Sec.~\ref{sec-revmes},
we briefly review quantum measurement theory.
We introduce the concept of unambiguous discrimination in Sec.~\ref{sec-uud}.
We present our main result on distinguishability 
in Sec.~\ref{sec-thm} and that on uniform distinguishability 
in Sec.~\ref{sec-thm-uni}. 
The excluded case in the main results, 
which is itself of interest, 
is addressed in Sec.~\ref{sec-perfect}. 
Section~\ref{sec-con} is for conclusion and discussions.

\section{Brief review of quantum measurement theory}
\label{sec-revmes}
In this section, we quickly review quantum measurement theory~\cite{DL70}
to the extent that is necessary for this paper. 
We consider the measurement with countably many outcomes.
%, unless otherwise stated. 
Let $\Omega$ denote the set of possible outcomes and 
$\hilh$ denote a system (a separable Hilbert space) to be measured. 
A state is expressed by a density operator $\rho\in \boneh$,
$\rho\geqslant0$, $\tr\rho=1$,
where $\boneh$ denotes the Banach space of trace class operators on $\hilh$.
A measurement on $\hilh$ with the 
outcome set $\Omega$ is mathematically described by
an \textit{instrument} $(\inst{I}_\omega)_{\omega\in\Omega}$.
Each element $\inst{I}_\omega$ of the instrument is
a linear map $\inst{I}_\omega:\boneh\to \boneh$ [that is bounded with
respect to the trace norm on $\boneh$], 
which describes a weighted state change caused by the measurement.
Each $\inst{I}_\omega$ 
sends a state 
$\rho$ 
to an ``unnormalized'' state, 
namely, %%$\inst I_\omega$ 
%%acts on the initial state $\rho$ as 
\begin{align}
 \inst{I}_\omega\rho
 =\Prob(\omega|\rho) \rho_\omega, 
 \label{eq-Irho}
\end{align}
where $\Prob(\omega|\rho)$ is 
the probability of obtaining an outcome $\omega\in\Omega$ 
and $\rho_\omega$ is the resulting state when the outcome 
$\omega$ is obtained. 
The equation \eqref{eq-Irho} 
defines $\Prob(\omega|\rho)$ and $\rho_\omega$ uniquely
(unless $\tr[\inst{I}_\omega\rho]=0$): 
  \begin{align}
   \Prob(\omega|\rho)
   &=\tr[\inst{I}_\omega\rho],&
   \rho_\omega
   &=\frac{\inst{I}_\omega\rho}{ \tr{[\inst{I}_\omega\rho]} }.
  \end{align}
  In order to interpret these quantities as probabilities and quantum
  state, respectively, 
  an instrument is assumed to satisfy the following two conditions.

(CP) Completely positivity: 
    $\inst{I}_\omega$ is completely positive for all $\omega\in\Omega$
    i.e., its trivial extentions 
    $\inst{I}_\omega\otimes \mathrm{id}_{\bC^{n\times n}}:
    \bonen(\hilh\otimes \bC^n)\to \bonen(\hilh\otimes \bC^n)$
    are positive maps for all $n\in\bN$.

(TP) Trace preserving property:
     The sum $\sum \inst{I}_\omega$ preserves trace of operators i.e.,
     $\tr[(\sum \inst{I}_\omega)\rho]=\tr \rho$ for all $\rho\in\bone{H}$.

The property CP, especially the positivity,
guarantees positivity of $\rho_\omega$ and $\Prob(\omega|\rho)$.
The property TP 
guarantees the conservation of probability. 
It is known that the instruments with above two properties correspond 
exactly 
to the realizable measurements (e.g.~\cite{doi:10.1063/1.526000}).

  \section{Unambiguous discrimination}
  \label{sec-uud}

  We begin with the precise definition of unambiguous discrimination. 
  
  \begin{defi}\label{defi-uud}
   An \emph{unambiguous discrimination measurement}
   %%with respect to states 
   between 
   $(\rho_j)_{j\in J}$ is
   an instrument $(\inst{I}_\omega)_{\omega \in J\cup\{\incon\}}$ that satisfies,
   for $j\ne k\in J$,
   \begin{align}
 \tr[\inst{I}_j\rho_k]
 &=0,&
 \tr[\inst{I}_j\rho_j]>0.
 \label{eq-def-uud}
   \end{align}
   We call $\tr[\inst{I}_j\rho_j]$ \emph{success probabilities}.
   The states $(\rho_j)_{j\in J}$ are said to be
   \emph{distinguishable}
   when they admit at least one unambiguous discrimination measurement
   between them.
  \end{defi}

  Let $(\rho_j)_{j\in J}$ be distinguishable
  candidate states. 
  Then 
  $(\rho_j)_{j\in J}$ 
  admit an unambiguous discrimination measurement
  $(\inst{I}_\omega)_{\omega \in J\cup\{\incon\}}$.
  When the true state is $\rho_j$,
  one obtains the outcome $j$ or {\incont}
  with probabilities $\tr[\inst{I}_j\rho_j]$ or $\tr[\inst{I}_\incon\rho_k]$, respectively, 
  and does not obtain any other outcome,
  %%that is, 
  namely, 
  $\tr[\inst{I}_j\rho_j]+\tr[\inst{I}_\incon\rho_j]=1$.
  When the outcome is $j$,
  we can decide the true state is $\rho_j$ \textit{with certainty}.
   
 In the rest of this section,
  we would like to discuss the 
  ways to quantify how good an unambiguous discrimination measurement is. 
  In the presence of a prior probability density $(p_j)_{j\in J}$, 
  which was assumed by 
  Dieks~\cite{Dieks1988303} and Chefles~\cite{Chefles1998339}, 
  it is reasonable to evaluate 
  an unambiguous discrimination measurement  $(\inst{I}_\omega)_{\omega\in J\cup\{\incon\}}$
  by the average success probability
  \begin{align}
   f_{\text{av}}
   :=\sum_{j\in J} p_jq_j,
   \label{eq-fav}
  \end{align}
  where $q_j=\tr[\inst{I}_j\rho_j]$ are success probabilities. 

  However, this is not the only way.
  %, as was mentioned in the Introduction.
  Even in the case of finite $J$, 
  there are important examples in which 
  $f$ is not of the form 
  \eqref{eq-fav}. 
  The minimum success probability 
  \begin{align}
   f_{\text{min}}
   :=\min\{\,q_j \mid j\in J\,\} 
   \label{eq-fmmin}
  \end{align}
  is such an example. 
  Roughly speaking, 
  $1/f_{\text{min}}$ %%times of  
  trials are enough to determine the
  true state. 
  This is the operational meaning of the minimum success probability
  $f_{\text{min}}$. 
  Generalizing these two examples,
  we define the class of evaluation functions for unambiguous
  discrimination. 

\begin{defi}\label{defi-evf}
 We call a function $f:(0,1]^J\to\bR$ an evaluation function if 
 \begin{align}
  x_j>y_j \ \text{for all}\  j\in J\implies f(x_j)>f(y_j)
  %\label{def-evfunc-inf}
 \end{align}
holds for all $(x_j)_{j\in J},\, (y_j)_{j\in J} \in (0,1]^J$.
We say an unambiguous discrimination measurement
$(\inst{I}_\omega)_{\omega \in J\cup\{\incon\}}$ 
between the states $(\rho_j)_{j\in J}$ is 
better if the value $f(\tr[\inst{I}_j\rho_j])$ is larger, 
and optimal if no other measurement exceeds the value.
%   Given an evaluation function $f:(0,1]^J\to\bR$,
%   we say that 
%   a
%   measurement 
%   $(\inst{I}_\omega)_{\omega \in J\cup\{\incon\}}$ 
%   provides 
%   better
%   unambiguous discrimination 
%   between 
%   the states $(\rho_j)_{j\in J}$
%   when the value $f(\tr[\inst{I}_j\rho_j])$
%   is larger.
%   We assume that the evaluation functions satisfy 
%   %%the following property.
%   \begin{align}
%  x_j>y_j \implies f(x_j)>f(y_j),
% \label{def-evfunc-inf}
%   \end{align}
%   where $(x_j)_{j\in J},\, (y_j)_{j\in J} \in (0,1]^J$.
  \end{defi}

  The class of ``evaluation functions'' defined here is so large
  that, 
  when $J$ is finite, 
  one could hardly imagine any functions 
  that suit the term 
  and do not belong to the defined class. 
  For example, 
  the class contains 
  $f_{\text{av}}$ and 
  $f_{\text{min}}$. 
%  Furthermore, it even contains 
%  a function such as 
%  $f_0(x_j):=x_0$, 
%  which favors the success probability for only one particular input
%  state $j=0$. 
  An evaluation function need not be linear nor convex.
  When $J$ is infinite, 
  %%We note, 
  however, 
  %%that 
  $f_{\text{inf}}$ [see \eqref{eq-finf}], 
  which is a natural generalization of 
  $f_{\text{min}}$,  
  is excluded from the class. 
  Such functions are more suitably discussed in the context of  
  uniform distinguishability 
  (see Sec.~\ref{sec-thm-uni}).

\section{Result on distinguishability}
\label{sec-thm}

  Now, we can state the first main result. 
  %%in this paper.
 \begin{theo}
  Optimal discriminations make
  distinguishable states indistinguishable
  when the discrimination fails (gives \incont). 
  Namely, assume that 
   the measurement 
   $(\inst{I}_{\omega})_{\omega\in J\cup \{\incon\}}$ 
   achieves an optimal 
   unambiguous discrimination between the states 
   $(\rho_j)_{j\in J}$ 
  and that the condition
  \begin{align}
   &\tr[\inst{I}_j\rho_j]<1 \quad\text{ for all $j\in J$}
   \label{theo-eq-assumption}
  \end{align}
  holds.
   Then the resulting states 
   under the condition that 
   the outcome is {\incont}, defined by 
   \begin{align}
 \left(\frac{\inst{I}_\incon\rho_j}{\tr[\inst{I}_\incon\rho_j]} \right)_{j\in J},
     \label{theo-eq}
   \end{align}
   are not distinguishable.
   Here,
   the optimality is defined
   by an arbitrary evaluation function 
   in Definiton~\ref{defi-evf}.
 \end{theo}

  Note that the condition $\tr[\inst{I}_j\rho_j]<1$ is equivalent to 
  $\tr[\inst{I}_?\rho_j]>0$ and this ensure the well-definedness of 
  the resulting states \eqref{theo-eq}.
  We will discuss the case 
  that 
  this condition fails in Sec.~\ref{sec-perfect}.

  We 
  %%shall first 
  explain the idea of the proof first 
  and then give the formal 
  one.
  The simple but important idea is to prove the contrapositive,
  namely, 
  if the states \eqref{theo-eq} 
  are distinguishable then the discrimination measurement 
  $(\inst{I}_\omega)_{\omega\in J\cup\{\incon\}}$ is not optimal.
  %%The proof is attributed to the construction of
  We thus construct 
  a discrimination measurement 
  that is better than the original one, 
  $(\inst{I}_\omega)_{\omega\in J\cup\{\incon\}}$, 
  %%by the assumption 
  assuming that \eqref{theo-eq} are distinguishable.
  %%We can carry out the task by considering a measurement consisting
  %%of the following two steps. 
  The measurement consisting of the following two steps does the
  task. 
  \begin{enumerate}
   \item
        Perform the original discrimination 
        $(\inst{I}_{\omega})_{\omega\in
        J\cup\{\incon\}}$. 
        If the outcome of this measurement is $j\in J$,
        then decide the true state is $\rho_j$
        regardless of the second step.
        If the outcome is {\incont},
        then defer the decision and proceed to the next step.
   \item
        Perform the discrimination of states \eqref{theo-eq},
        whose existence is guaranteed by the assumption.
        If the outcome of this measurement is $j\in J$,
        decide the true state is $\rho_j$.
        Otherwise, give up on the decision and answer {\incont}.
  \end{enumerate}
%%   In short, if the resulting states \eqref{theo-eq} are still
%%   distinguishable, we can perform the discrimination between  them and obtain more accurate discrimination.
  %%We can carry out the task by considering a measurement consisting
  %%of the following two steps. 
  We show below
  that this combined discrimination measurement
  truly improves the original one 
  $(\inst{I}_{\nu})_{\nu\in J\cup\{\incon\}}$.
%%  Thanks to the appropriate preparation in the previous section, 
%%   It is seen that the definitions in the previous section 
%%   works properly in the proof. 
  \begin{proof}[Proof of the Theorem]
   Note that
   \begin{align}
 \tr[\inst{I}_?\rho_j]
 &= 1-\tr[\inst{I}_j\rho_j]
 >0
   \end{align}
   by the assumption of the Theorem.
   
   We will prove the contrapositive. 
   %%of the claim.
   Let us assume \eqref{theo-eq} admits an unambiguous discrimination
   measurement $(\inst{I}'_\omega)_{\omega\in \Omegaa}$. 
   %%where we
   %%denote $J\cup \{\incon\}$ by $\Omega$. 
   Because $(\inst{I}_{\omega})_{\omega\in\Omegaa}$ and $(\inst{I'}_{\omega})_{\omega\in\Omegaa}$
   discriminate between $(\rho_j)_{j\in J}$ and between \eqref{theo-eq}, respectively,
   one obtains, by Definition~\ref{defi-uud},
   \begin{align}
 &\tr\left[\inst{I}_j\rho_k\right]
 =0,&
 & \tr\left[\inst{I}_j\rho_j\right]>0,
 \label{theo-pr-eq1}\\
 &\tr\left[\inst{I}_j
 \frac{\inst{I}_\incon\rho_k}{\tr[\inst{I}_\incon\rho_k]}
 \right]
 =0,&
 &
 \tr\left[\inst{I}_j
 \frac{\inst{I}_\incon\rho_j}{\tr[\inst{I}_\incon\rho_j]}
 \right]
 >0.
 \label{theo-pr-eq2}
   \end{align}
   for $j, k\in J$ with $j\ne k$.
   
   Let us define an instrument 
   $(\inst{J}_{\omega})_{\omega\in \Omegaa}$ 
   by
%%   \begin{align}
%%    \inst{J}_{\omega}
%%    =
%%    \begin{cases}
%%  \big(\sum_{\omega'\in\Omegaa}\inst{I}'_{\omega'}\big) \inst{I}_\omega+ \inst{I}'_{\omega} \inst{I}_?, &
%%  \omega\in J,\\ 
%%  \inst{I}'_\incon\inst{I}_\incon, 
%%  &
%%  \omega=\,\incon, 
%%    \end{cases}
%%    \label{theo-pr-eq3}
%%   \end{align}
  \begin{align}
   \inst{J}_j
   &:=
   \bigg(\sum_{\omega'\in\Omegaa}\inst{I}'_{\omega'}\bigg) \inst{I}_j+ \inst{I}'_{j} \inst{I}_?, &
 j\in J,\\ 
   \inst{J}_\incon&:=
 \inst{I}'_\incon\inst{I}_\incon, 
   \label{theo-pr-eq3}
  \end{align}
  We will prove that
  the instrument $(\inst{J}_{\omega})_{\omega\in\Omegaa}$ unambiguously discriminate
  states $(\rho_j)_{j\in J}$ better than $(\inst{I}_\omega)_{\omega\in\Omegaa}$
   in the rest of this proof.

  First, it is readily seen that 
  $ (\inst{J}_{\omega})_{\omega\in \Omegaa}$ is an instrument
  since
  %$(\inst{I}_{\omega})_{\omega\in\Omegaa}$ and $(\inst{I}_{\omega}')_{\omega\in\Omegaa}$ are instruments and
  $\sum_{\omega\in \Omegaa} \inst{J}_\omega=\sum_{\omega',\omega\in \Omegaa}
  \inst{I}'_{\omega'}\inst{I}_{\omega} $.

   Second,
   we calculate the probabilities $\tr\left[ \inst{J}_{j}\rho_k \right]$
   for all $j,k\in J$:
  \begin{align}
   &\tr\left[ \inst{J}_{j}\rho_k \right]\notag\\
   &=
   \tr\left[
   \left(  \sum \inst{I}'_{\omega} \right) \inst{I}_j \rho_k
   \right]
   +
   \tr\left[
   \inst{I}'_{j}\inst{I}_\incon \rho_k
   \right]
   \notag\\
   &=
   \tr\left[ \inst{I}_j \rho_k \right]
   +
   (\tr[ \inst{I}_\incon \rho_k ])
   \tr\left[
   \inst{I}'_{j} \frac{\inst{I}_\incon \rho_k}{ \tr[\inst{I}_\incon \rho_k] }
   \right]
   \notag\\
   &=
   \delta_{j,k}  \left(
   \tr\left[ \inst{I}_k \rho_k \right]
   +
   (\tr[ \inst{I}_\incon \rho_k ])
   \tr\left[
   \inst{I}'_{k} \frac{\inst{I}_\incon \rho_k}{ \tr[\inst{I}_\incon \rho_k] }
   \right]
   \right),
   \label{theo-rp-succpr}
  \end{align}
 where the first equality is by 
 %%\eqref{theo-pr-eq3}, 
 the definition of 
 $ (\inst{J}_{\omega})_{\omega\in \Omegaa}$, 
 the second follows from the TP property of $\sum \inst{I}'_{\omega}$,
 %of the instrument $(\inst{I'}_{\omega})_{\omega\in\Omegaa}$,
 and the third is by \eqref{theo-pr-eq1} and \eqref{theo-pr-eq2}.

 Finally,  we show that $(\inst{J}_{\omega})_{\omega\in\Omegaa}$
  unambiguous discriminates between $(\rho_j)_{j\in J}$ better than $(\inst{I}_\omega)_{\omega\in \Omegaa}$.
  %%In order to show this,
  %%we must 
  %%compare the success probabilities $\tr[\inst{J}_j\rho_j]$ and $
  %%\tr[\inst{I}_j\rho_j]$: 
  We see that each 
  $\tr[\inst{J}_j\rho_j]$ is 
  strictly larger than 
  $\tr[\inst{I}_j\rho_j]$: 
   \begin{align}
 &\tr[\inst{J}_j\rho_j] - \tr[\inst{I}_j\rho_j]
 %%\notag\\
 =
 \tr[ \inst{I}_\incon \rho_j ]
 \tr\left[
 \inst{I}'_{j} \frac{\inst{I}_\incon \rho_j}{ \tr[\inst{I}_\incon \rho_j] }
 \right]
 %%\notag\\
 >0,
 \label{theo-eq-posi}
   \end{align}
 where the equality is by \eqref{theo-rp-succpr},
 and the inequality is 
 by \eqref{theo-pr-eq1} and \eqref{theo-pr-eq2}.
%%    This in particular 
%%    guarantees that $(\inst{J}_\omega)_{\omega\in \Omegaa}$
%%    unambiguously discriminate between 
%%    $(\rho_j)_{j\in J}$, 
%%    since $\tr\left[ \inst{J}_{j}\rho_j \right]>0$ holds by \eqref{theo-rp-succpr}.
   %
   Eqs.~\eqref{theo-rp-succpr} and \eqref{theo-eq-posi} prove, in particular, 
   that $(\inst{J}_\omega)_{\omega\in \Omegaa}$
   unambiguously discriminate $(\rho_j)_{j\in J}$.
   %since $\tr\left[ \inst{J}_{j}\rho_j \right]>0$ holds by \eqref{theo-rp-succpr}.
   Let $f$ be any evaluation function described in the Definition~\ref{defi-evf}.
   Then, by \eqref{theo-eq-posi},
   we have $f(\inst{J}_j\rho_j)> f(\inst{I}_j\rho_j)$.
   In other words,
   the instrument 
   $(\inst{J}_{\omega})_{\omega\in\Omegaa}$ 
   discriminates between 
   $(\rho_j)_{j\in J}$ 
   better than $(\inst{I}_\omega)_{\omega\in \Omegaa}$. 
   This completes the proof.
  \end{proof}

  We will discuss 
  the assumption \eqref{theo-eq-assumption} in the Theorem 
  in Sec.~\ref{sec-perfect}.

\section{Result on uniform distinguishability}
\label{sec-thm-uni}
  As the number of states becomes infinite, the concept of distinguishability
  is naturally divided into two: 
  ``distinguishability'' and ``uniform distinguishability''~\cite{1751-8121-49-26-265201}. 
  We discussed the former 
  in the preceding sections. 
  We consider the latter in this section.

  We recall 
  that the class of evaluation functions
  in the Definition in Sec.~\ref{sec-uud}
  becomes slightly restricted
  when 
  the index set $J$ is not finite.
  Although $f_{\text{av}}$ 
  %%and $f_0$ are 
  is 
  included in the class, 
  an important 
  evaluation function 
  \begin{align}
   f_{\text{inf}}(x_j)
   :=\inf\{\, x_j\mid j\in J\,\}, 
   \label{eq-finf}
  \end{align}
  which generalizes 
  $f_{\text{min}}$, %%(x_j)$ 
  is excluded. 
  The function 
  $f_{\text{inf}}$
  has a definite operational meaning 
  similar to 
  $f_{\text{min}}$. 
  Hence, it is 
  a natural demand to include 
  such an evaluation function. 
  It can be done 
  by introducing the uniform distinguishability.  
  We provide the uniform version of
  Definitions~\ref{defi-uud} and \ref{defi-evf}, 
  and the Theorem. 
  
  \begin{defi-a}\label{defi-uud-unif}
   A \emph{uniform unambiguous discrimination measurement}
   between 
   $(\rho_j)_{j\in J}$ is
   an instrument $(\inst{I}_\omega)_{\omega \in J\cup\{\incon\}}$ that satisfies,
   for $j\ne k\in J$,
   \begin{align}
    \tr[\inst{I}_j\rho_k]
    &=0,&
    \inf\{\,\tr[\inst{I}_j\rho_j] \mid j\in J\,\}>0.
    \label{eq-def-uud}
   \end{align}
   The states $(\rho_j)_{j\in J}$ are said to be
   \emph{uniformly distinguishable}
   when they admit at least one 
   uniform 
   unambiguous discrimination measurement
   between them.
  \end{defi-a}

  Distinguishability is a weaker condition than
  uniform distinguishability. 
  Consider, for example, 
  the states $(\rho_j)_{j\in\bN}$ that 
  are merely distinguishable 
  with success probabilities 
  $q_j=\tr[\inst{I}_j\rho_j]=1/j$.
  In this case, 
  one cannot predict how many trials suffices 
  to determine the true state 
  before performing the discrimination
  %%in advance 
  since $1/f_\text{inf}(q_j)=1/0=\infty$.
  Uniform distinguishability 
  cures this problem 
  and 
  provides a natural framework for infinitely many states.

  \begin{defi-a}\label{defi-evf-unif}
  We call a function $f:(0,1]^J\to\bR$ 
  an evaluation function \emph{for uniform discrimination} if 
   \begin{align}
    \inf\{\, x_j>y_j  \mid j\in J\,\} >0\implies f(x_j)-f(y_j)>0,
    \label{def-evfunc-inf}
   \end{align}
   holds, where $(x_j)_{j\in J},\, (y_j)_{j\in J} \in (0,1]^J$.
  \end{defi-a}

 Note that the function $f_\text{inf}$ is an evaluation function for
 uniform discrimination as well as $f_\text{av}$. 
 %%and $f_0$. 
 Evaluation functions for uniform discrimination 
 comprise a 
 larger  class than mere evaluation functions do.

 \begin{theo-a}
  Optimal uniform discrimination measurements make
  uniformly distinguishable states not uniformly distinguishable 
  when the discrimination fails (gives \incont). 
  Namely, assume that 
   the measurement 
   $(\inst{I}_{\omega})_{\omega\in J\cup \{\incon\}}$ 
   achieves an optimal uniform
   unambiguous discrimination between the states 
   $(\rho_j)_{j\in J}$ 
  and that the condition
  \begin{align}
   &\sup\{\tr[\inst{I}_j\rho_j]\}<1 
   \label{theo-eq-assumption-a}
  \end{align}
  holds.
   Then the resulting states 
   under the condition that 
   the outcome is {\incont}, defined by 
   \begin{align}
 \left(\frac{\inst{I}_\incon\rho_j}{\tr[\inst{I}_\incon\rho_j]} \right)_{j\in J},
    %\label{theo-a-eq}
   \end{align}
   are not uniformly distinguishable.
   Here,
   the optimality is defined
   by an arbitrary evaluation function 
   for uniform discrimination 
   in Definition~\ref{defi-evf}\,$'$.
 \end{theo-a}

 The definitions and theorem with prime symbols
 are equivalent to those without prime in Sec.~\ref{sec-thm} when $J$
 is finite. 
 When $J$ becomes countably infinite,
 ``uniform distinguishability'' becomes a stronger condition than mere
 ``distinguishability'' and evaluation functions for uniform
 discrimination form a larger class than that of mere evaluation 
 functions.  
 The proof of the Theorem$'$ can be given in a way similar to that of
 the Theorem  
 and is omitted [the essential point is
 to replace the inequalities ``$\cdots>0$'' with ``$\inf\{\,\cdots\mid
 j\in J\,\}>0$'' 
 in 
 Eqs.~\eqref{theo-pr-eq1}, 
 \eqref{theo-pr-eq2} and 
 \eqref{theo-eq-posi}].

 We note 
 that an assumption \eqref{theo-eq-assumption-a}, 
 which is stronger  
 than \eqref{theo-eq-assumption}  
 in the previous Theorem, 
 is necessary in the Theorem$'$. 
 In fact, 
 when 
 $\sup q_j=1$ and $q_j<1$, 
 the original uniform discrimination 
 (with success probabilities $q_j$) 
 is not 
 necessarily improved by a 
 subsequent uniform discrimination measurement 
 (with $q_j'$). 
 The reason is 
 because the improvements in success probabilities are given by
 $(1-q_j)q_j'$ [see \eqref{theo-eq-posi}].

  \section{Separation of perfectly distinguishable states}
  \label{sec-perfect}

  In the Theorem, we assumed that 
  optimal discrimination measurement 
  satisfies $\tr[\inst{I}_j\rho_j]<1$.
  The excluded case was $\tr[\inst{I}_\incon\rho_j]=0$,
  or equivalently, 
  $\tr[\inst{I}_j\rho_j]=1$.
  We discuss such cases in this section. 
  \begin{defi}
    Let $(\rho_j)_{j\in J}$ be states and 
   $K$ be a subset of $J$. 
    The 
   states $(\rho_j)_{j\in K}$
    are said to be 
    perfectly distinguishable 
   if there exists 
   an 
   unambiguous discrimination measurement 
   $(\inst{I}_{\omega})_{\omega \in\jincon}$ 
   between 
   $(\rho_j)_{j\in J}$ 
   such that 
    \begin{align}
     &\tr[\inst{I}_k\rho_k]=1 
    \end{align} 
    holds for all $k\in K$.
 \end{defi}
  \begin{prop}
   Assume 
   that 
   the 
   states $(\rho_j)_{j\in J}$ are 
   distinguishable and 
   that $(\rho_j)_{j\in K},\, K\subset J$, are
   perfectly distinguishable.
   Then there exists a (two-outcome projection) measurement 
   $(\inst{L},\inst{L}')$ such that,
   \begin{align}
    &\inst{L}\rho_k=\rho_k, &
    &\inst{L}'\rho_k=0,&
    &k\in K,
    \label{eq-prop-1}\\
    \intertext{and} 
    &\inst{L}\rho_\ell=0,&
    &\inst{L}'\rho_\ell=\rho_\ell,&
    &\ell\in J\setminus K
    \label{eq-prop-2}
   \end{align} 
   holds.
  \end{prop}

\begin{proof} %%[Proof of Proposition]
 Assume $(\rho_k)_{k\in K}$ are perfectly distinguishable by an unambiguous discrimination measurement $(\inst{I}_\omega)_{\omega\in J\cup\{\incon\}}$. 
    Let $L\in\bofh$ be an operator such that, for all $\rho\in\bone{H}$,
    \begin{align}
     \sum_{k\in K} \tr[\inst{I}_k\rho]
     =\tr[LL^*\rho],
    \end{align}
    i.e., 
    $LL^*$ is the sum of so-called positive-operator valued measure
    (POVM) elements for outcomes in $K$. 
    Then $LL^*\leqslant1$. 
    %% by definition.
    Let $P$ be the projection onto the (norm) closure of 
    $L\hilh=\{\, L\xi \mid \xi \in\hilh \,\}$. 
    Then $PL=L$. %% by definition.
    Define the instrument $(\inst{L},\inst{L}')$ by
    \begin{align}
     \inst{L}\rho&=P\rho P,&
     \inst{L}'\rho&=(1-P)\rho(1-P),
    \end{align}
    where $\rho\in\boneh$.
    We will prove this instrument satisfies 
    the conditions \eqref{eq-prop-1} and \eqref{eq-prop-2} in the Proposition.

 First, we prove \eqref{eq-prop-1}.
 Fix $k\in K$.
 Because $(\inst{I}_{\omega})_{\omega \in\jincon}$ is perfect on $K$, we have
 \begin{align}
  1
  &=\tr[\inst{I}_k\rho_k] 
  \notag \\
  &
  \leqslant \sum_{k'\in K}\tr[\inst{I}_{k'}\rho_k] 
  %%\notag\\
  %%&
  =\tr[LL^*\rho_k] %%\notag 
  %%&
  =\tr[LL^*(P\rho_k P)] \notag \\
  &\qquad\qquad\qquad\qquad\qquad\qquad\quad
  \leqslant \tr[1(P\rho_k P)].
 \end{align}
 Therefore $\tr[P\rho_kP]=1$.
 Thus, by the cycric property of trace, $\tr[(1-P)\rho_k(1-P)]=0$.
 Then one has $\rho_k^{1/2}(1-P)=0$,
 which proves \eqref{eq-prop-1}.
 
 Second, we prove \eqref{eq-prop-2}.
 Fix $\ell\in J\setminus K$.
 By the assumption that $(\inst{I}_\omega)_{\omega\in \jincon}$
 unambiguously discriminates between the states,
 we have $\tr[\inst{I}_k \rho_\ell]=0$ for all $k\in K$.
 Then $0=\sum_{k\in K} \tr[\inst{I}_k\rho_\ell]=\tr[L^*\rho_\ell L]$ 
 and $\rho_\ell^{1/2}L=0$.
 Since $P$ is the projection onto the closure of $L\hilh$,
 we have $\rho_\ell^{1/2}P=0$.
 Thus $(1-P)\rho_\ell(1-P)=\rho_\ell$,
 which proves \eqref{eq-prop-2}.
\end{proof}

By the measurement described in the Proposition,
we can see whether the true state $\rho_j$ 
is in 
$(\rho_k)_{k\in K}$ or $(\rho_\ell)_{\ell\in J\setminus K}$, 
without disturbing 
the set that contains $\rho_j$. 
   
The Theorem 
in Sec.~\ref{sec-thm} 
is not applicable when an optimal instrument perfectly
discriminates between some of the states.   
In such a case, 
one can remove all perfectly distinguishable states
beforehand by the 
Proposition above 
and 
then apply the Theorem.
Therefore, 
we can assume \eqref{theo-eq-assumption} in the Theorem 
without loss of generality.
%%
%% Similarly, 
%% we can assume \eqref{theo-eq-assumption} in the
%% discussion on uniform distinguishability. 
However, 
we cannot always assume
%%the stronger assumption 
\eqref{theo-eq-assumption-a} in the
Theorem$'$ 
%%[which is stronger than \eqref{theo-eq-assumption}]
physically. 
%%cannot be removed completely. 
%%the case 
%%of $\sup q_j=1$ with $q_j<1$, 
%%where $q_j$ are the success probabilities. 
When $q_j:=\tr[\inst{I}_j\rho_j]<1$ and $\sup q_j=1$, 
the states after an optimal uniform discrimination measurement with 
the inconclusive outcome may be uniformly distinguishable.

  \section{Conclusion and discussions}
  \label{sec-con}

We have shown that optimal unambiguous discrimination makes
distinguishable states indistinguishable under the condition that the
inconclusive outcome {\incont} is obtained. The results extend the
previously known ones to infinitely many candidates and output states,
which can be pure or mixed, to all quantum mechanically possible
measurements, and to virtually all evaluation functions that define
optimality. Our proof was based on a simple principle. The best
measurements leave no room to carry out the task further. If the
resulting states are distinguishable, we can achieve more accurate
discrimination by discriminating the resulting states. This made the
proof almost obvious and, at the same time, removed restrictions in
the previous work~\cite{Chefles1998339}. We would like to emphasize
that our proofs of the Theorems do not depend on the criteria of the
distinguishability
%%, which characterize the distinguishability with terms 
such as linear independence. The Theorems is a direct
consequence of the definition of optimality.  

We have also discussed the uniform discrimination, which becomes slightly
stronger than the mere distinguishability when the number of candidate
state becomes infinite. We showed the Theorem$'$ for the uniform
distinguishability. The conclusion of the Theorem$'$ becomes 
slightly weaker; however, the class of evaluation functions is
enlarged, which includes a natural measure $f_{\text{inf}}$. The
difference between the Theorem and Theorem$'$ describe the subtlety 
of 
the handling of infinite many states. 

Besides the main results, we have discussed the case that some
candidate states are perfectly distinguishable. We showed that such
candidate state can be separated by the two-outcome measurement
without disturbing the true state. This 
%%seems interesting and
reflects the nature of unambiguous discrimination.  

If one is interested in more detailed descriptions of the property of
the resulting states, he/she can make use of the results in our
previous work. %%In the work, we 
It was 
proved that countably many pure states
are distinguishable if and only if they are minimal and that they are
uniformly distinguishable if and only if they are Riesz-Fisher (both
of the mathematical properties are generalizations of the linear
independence to the case of infinitely many vectors). We also derived
the condition for the countably many general (possibly mixed) states
to be distinguishable (see also \cite{PhysRevA.70.012308} for the case
of finitely many states). On the other hand, the condition for
countably many general states to be \textit{uniformly} distinguishable
can be stated based on our previous work in principle; however, the
condition so derived seems to be complicated so that it is not
practical enough. It may be an interesting problem to simplify the
condition for uniform distinguishability, which 
%%also 
helps understand
the resulting states or the disturbance of the optimal
discrimination. 

Finally, we would like to make a general comment on the simple
principle which we have used to prove the theorems (that the optimal
measurements leave no room to carry out the task further). The idea is
itself very simple and obvious so that everyone understands it
readily. However, it is not trivial in what context this idea really
works well and how to apply the idea to each context. Note that the
fact that the idea works in the context of unambiguous discrimination
under an appropriate setting, as we have demonstrated in this paper,
is itself not trivial. For example, without the inconclusive outcome,
it might be impossible to make use of the idea in 
%%the present context. 
a similar manner. 
The simple idea seems to be applicable to a wide variety of
subjects 
%%in quantum foundations and information 
and also allows a
unified discussion. It would be an interesting work to find other
subjects 
%%that can be generalized or proved in a simple manner based on
%%such an idea. 
%%the idea. 
in which the idea can 
draw useful conclusions.


\begin{thebibliography}{99}


%%%%%%%%%%%%%%%%%%%%%%%%%%%%%%%%%%Error-DisturbanceRelations
\bibitem{Ozawa04}%
For modern treatments,
see 
M.~Ozawa,
Phys. Rev. A 
\text{67} 
042105
(2004);
%
%\bibitem{WSU11}%
Y.~Watanabe, 
T.~Sagawa, and
and M.~Ueda
Phys. Rev. A 
\textbf{84} 
042121
(2011).


%%%%%%%%%%%%%%%%%%%%%%%%%%%%%%%%%%QuantumCodingTheorems
\bibitem{Hol98}
See 
A.~S.~Holevo,
Russ. Math. Surv.
\textbf{53} 
1295
(1998)
and reference therein.


%%%%%%%%%%%%%%%%%%%%%%%%%%%%%%%%%%EntanglementDistillation
\bibitem{Bennett96a}%
C.~H.~Bennett, 
G.~Brassard,  
S.~Popescu,
B.~Schumacher,
J.~A.~Smolin, and
W.~K.~Wootters,
Phys. Rev. Lett.
\textbf{76} 
722 
(1996); 
%%
%%\bibitem{Bennett96b}%
C.~H.~Bennett, 
H.~J.~Bernstein, 
S.~Popescu, and 
B.~Schumacher,
Phys. Rev. A 
\textbf{53} 
2046 
(1996); 
%%
%%\bibitem{Bennett96c}%
C.~H.~Bennett, 
D.~P.~DiVincenzo, 
J.~A.~Smolin, and
W.~K.~Wootters,
Phys. Rev. A 
\textbf{54} 
3824 
(1996).


%%%%%%%%%%%%%%%%%%%%%%%%%%%%%%%%%%SateEstimation
\bibitem{Helstrom76}%
For example,
C.~W.~Helstrom,
\textit{Quantum Detection and Estimation Theory}
(Academic Press, New York, 1976).


%%%%%%%%%%%%%%%%%%%%%%%%%%%%%%%%%%StateDiscrimination

\bibitem{Chefles00}%
For an overview, see 
A.~Chefles,
Contemp. Phys.
\textbf{41}
401
(2000). 
%%and references therein. 


%%%%%%%%%%%%%%%%%%%%%%%%%%%%%%%%%%state protection
\bibitem{Wakamura17a}%
For exmaple, 
H.~Wakamura, 
R.~Kawakubo, 
and T.~Koike,
Phys. Rev. A 
\textbf{95} 
022321  
(2017); 
%%\bibitem{Wakamura17b}%
%%H.~Wakamura, 
%%R.~Kawakubo, 
%%and T.~Koike,
Phys. Rev. A 
\textbf{96} 
022325
(2017).



%%%%%%%%%%%%%%%%%%%%%%%%%%%%%%%%%%UD
\bibitem{Ivanovic1987257}%
I.~Ivanovic,
Phys. Lett. A
\textbf{123},
257
(1987).
\bibitem{Dieks1988303}%
D.~Dieks,
Phys. Lett. A
\textbf{126},
303
(1988).
\bibitem{Peres198819}%
A.~Peres,
Phys. Lett. A
\textbf{128},~
19
(1988).
\bibitem{Jaeger199583}%
G.~Jaeger, and A.~Shimony,
Phys. Lett. A
\textbf{197},
83
(1995).
\bibitem{Chefles1998339}%
A.~Chefles,
Phys. Lett. A
\textbf{239},
339
(1998).
\bibitem{PhysRevA.70.012308}%
Y.~Feng,~R.~Duan, and M.~Ying,
Phys. Rev. A
\textbf{70},
012308
(2004).


\bibitem{vonneumann32}%
J. von Neumann. “Mathematische Grundlagen der Quantenmechanik.”
Springer Berlin
Heidelberg, 1932.


\bibitem{1751-8121-49-26-265201}%
R.~Kawakubo, and T.~Koike, 
J. Phys. A: Math. Theor. 49, 
265201 
(2016).


\bibitem{DL70}%
E.~B.~Davies, and 
J.~T.~Lewis, 
Commun. Math. Phys. 
\textbf{17}, 
239 
(1970). 

\bibitem{doi:10.1063/1.526000}%
M.~Ozawa,
J. Math. Phys. 
\textbf{25},
79
(1984).
\end{thebibliography}
\end{document}